\documentclass[aps,prx,10pt,twocolumn,superscriptaddress,nofootinbib,longbibliography]{revtex4-2}

\usepackage[T1]{fontenc}
\usepackage{lmodern}
\usepackage{microtype}
\usepackage{amsmath,amssymb,amsthm,mathtools}
\usepackage{booktabs}
\usepackage{enumitem}
\usepackage{float}
\usepackage{tikz}
\usepackage{placeins}
\usepackage{needspace}
\usepackage{xcolor}
\usepackage[hidelinks]{hyperref}
\usepackage[nameinlink,capitalise]{cleveref}

\hypersetup{pdftitle={Quantum Approximation Complexity of Classical Optimization Problems},pdfauthor={Stuart Hadfield}}

\newtheorem{theorem}{Theorem}
\newtheorem{lemma}[theorem]{Lemma}
\newtheorem{proposition}[theorem]{Proposition}
\newtheorem{corollary}[theorem]{Corollary}
\newtheorem{definition}[theorem]{Definition}
\newtheorem{remark}[theorem]{Remark}

\newcommand{\APX}{\mathsf{APX}}
\newcommand{\PTAS}{\mathsf{PTAS}}
\newcommand{\FPTAS}{\mathsf{FPTAS}}
\newcommand{\RAPX}{\mathsf{R\text{-}APX}}
\newcommand{\RPTAS}{\mathsf{R\text{-}PTAS}}
\newcommand{\RFPTAS}{\mathsf{R\text{-}FPTAS}}
\newcommand{\BQAPX}{\mathsf{BQ\text{-}APX}}
\newcommand{\BQPTAS}{\mathsf{BQ\text{-}PTAS}}
\newcommand{\BQFPTAS}{\mathsf{BQ\text{-}FPTAS}}
\newcommand{\RLOGAPX}{\mathsf{R\text{-}LOG\text{-}APX}}
\newcommand{\RPOLYAPX}{\mathsf{R\text{-}POLY\text{-}APX}}
\newcommand{\REXPAPX}{\mathsf{R\text{-}EXP\text{-}APX}}
\newcommand{\BQLOGAPX}{\mathsf{BQ\text{-}LOG\text{-}APX}}
\newcommand{\BQPOLYAPX}{\mathsf{BQ\text{-}POLY\text{-}APX}}
\newcommand{\BQEXPAPX}{\mathsf{BQ\text{-}EXP\text{-}APX}}
\newcommand{\RAPTAS}{\mathsf{R\text{-}APTAS}}
\newcommand{\FBPP}{\mathsf{FBPP}}
\newcommand{\BQP}{\mathsf{BQP}}
\newcommand{\NP}{\mathsf{NP}}
\newcommand{\OPT}{\operatorname{OPT}}
\newcommand{\ord}{\operatorname{ord}}
\newcommand{\CMO}{\textnormal{\textsc{CMO}}}
\newcommand{\DLog}{\textnormal{\textsc{DiscreteLog}}}
\newcommand{\Factor}{\textnormal{\textsc{Factor}}}
\newcommand{\DLogFit}{\textnormal{\textsc{DLog-Fit}}}
\newcommand{\Tr}{\operatorname{Tr}}

\begin{document}
\raggedbottom
\title{Quantum Approximation Complexity of Classical Optimization Problems}
\author{Stuart Hadfield}
\email{shadfield@usra.edu}
\thanks{ORCID: \href{https://orcid.org/0000-0002-4607-3921}{0000-0002-4607-3921}.}
\affiliation{USRA Research Institute for Advanced Computer Science, Mountain View, California 94043, USA}
\date{October 5, 2026}
\begin{abstract}
Classical approximation complexity asks what solution quality can be guaranteed with polynomial-time computation.  The classes \(\APX\), \(\PTAS\), and \(\FPTAS\) distinguish a fixed approximation ratio, approximation to any fixed accuracy, and approximation schemes whose running time is also polynomial in inverse accuracy.  Their randomized counterparts are \(\RAPX\), \(\RPTAS\), and \(\RFPTAS\).  We define bounded-error quantum counterparts \(\BQAPX\), \(\BQPTAS\), and \(\BQFPTAS\).  Membership requires a uniform quantum algorithm that, on every input, returns a feasible classical solution achieving at least the claimed approximation ratio with probability at least \(2/3\).  The solution's score (objective value) must be efficiently classically computable.  Running time includes parameter selection, preparation, measurement, decoding, and repetition.  Many quantum optimization methods are used heuristically, and high benchmark scores alone do not establish these guarantees.  We further establish a conditional hierarchy for logarithmic, polynomial, and exponential approximation factors.

Assuming \(\NP\nsubseteq\BQP\), the quantum classes form a strict hierarchy.  Problems based on prime factorization and discrete logarithms give conditional quantum--classical separations.  Certified Maximum Order has an exact quantum algorithm, while any randomized polynomial-time algorithm guaranteeing at least an inverse-polynomial approximation ratio on every input would yield efficient factoring.  Discrete-Logarithm Fitting has an exact quantum algorithm and a deterministic one-half approximation, but any fixed improvement over one half would give a randomized polynomial-time algorithm for the safe-prime discrete logarithm problem.  Our results show, under explicit complexity assumptions, that quantum computation can improve worst-case approximation guarantees.  A quantum--classical gap in worst-case output ratios for common problems such as MaxCut or MaxSAT remains open.
\end{abstract}
\maketitle

\section{Introduction}
\label{sec:intro}

Quantum optimization asks whether quantum computers can provide advantages in finding classical solutions to hard combinatorial problems~\cite{Abbas2024}.  Despite much effort, many quantum algorithms remain heuristic and lack rigorous performance guarantees across all inputs or bounds on total runtime.  \emph{Approximation complexity} classifies optimization problems by the solution quality that algorithms can guarantee with polynomial-time computation.  A solution's score is its objective value.  For maximization, the achieved approximation ratio is the score divided by the optimum.  The classical classes \(\APX\), \(\PTAS\), and \(\FPTAS\) contain problems with increasingly demanding polynomial-time approximation guarantees~\cite{PapadimitriouYannakakis1991,Ausiello1999}.  A problem is in \(\APX\) if a deterministic polynomial-time algorithm achieves some fixed approximation ratio on every instance.  It is in \(\PTAS\) if one scheme achieves any requested accuracy in polynomial time for each fixed accuracy level.  It is in \(\FPTAS\) if the running time is also polynomial in inverse accuracy.  For example, MaxCut and Max-E3SAT are in \(\APX\), Planar Vertex Cover is in \(\PTAS\), and 0--1 Knapsack is in \(\FPTAS\)~\cite{GoemansWilliamson1995,Hastad2001,Baker1994,IbarraKim1975}.  A glossary of the complexity classes appears in Appendix~\ref{app:class-glossary}.

Quantum optimization routines can produce candidate solutions with high scores on simulated or experimental test instances~\cite{KadowakiNishimori1998,Farhi2014,Hadfield2019,MarshWang2019,MaciejewskiNDAR2025,MaciejewskiNDAWS2026,LotshawIWS2026,Sciorilli2025}.  Such scores alone establish neither a worst-case ratio nor a bound on the cost of parameter search and sampling~\cite{Abbas2024}.  The quantum approximate optimization algorithm (QAOA)~\cite{Farhi2014,Hadfield2019} also has rigorous guarantees in specific settings.  Analytic depth-one MaxCut formulas and fixed-angle bounds for unweighted 3-regular graphs give rigorous bounds on expected approximation ratios~\cite{Farhi2014,WangHadfieldJiangRieffel2018,WurtzLove2021}.

Assessing a worst-case approximation advantage requires a uniformly specified algorithm, a guaranteed ratio for the returned solution, and a bound on the full computational cost.  We work with NP optimization (NPO) problems, whose candidate solutions have polynomial length and whose feasibility and scores are checkable in polynomial time.  We extend the classical approximation classes to bounded-error quantum algorithms on classical inputs.  On every input, such an algorithm must return a feasible classical solution meeting the claimed ratio with probability at least \(2/3\).  All quantum and classical algorithmic steps count toward running time.  We use the prefix ``BQ'' to emphasize bounded error and avoid the established use of \(\mathsf{QPTAS}\) for a quasi-polynomial-time approximation scheme.  We ask whether the deterministic, randomized, and bounded-error quantum versions of each class coincide and, if not, which optimization problems separate them.  A problem in a BQ class but outside its randomized counterpart would establish a conditional quantum advantage in guaranteed solution quality or in the running time needed to achieve a requested accuracy.

Under \(\NP\nsubseteq\BQP\), known classical approximation algorithms and hardness results imply that Max-E3SAT has a deterministic constant-ratio approximation but no quantum PTAS, while Planar Vertex Cover has a deterministic PTAS but no quantum FPTAS.  Hence these problems witness the strict hierarchy in \Cref{prop:hierarchy},
\[
\BQFPTAS\subsetneq\BQPTAS\subsetneq\BQAPX.
\]
These examples establish strict inclusions among the quantum classes without separating a quantum class from its randomized classical counterpart.

Certified Maximum Order (CMO) gives our strongest conditional quantum--classical separation.  On any integer \(N\ge3\), a feasible solution certifies the multiplicative order of an element modulo \(N\), and its score is that order.  The optimum is the Carmichael function \(\lambda(N)\), the largest order of an element coprime to \(N\).  A quantum algorithm using factoring and period finding constructs an exact optimum~\cite{Shor1997}.  Conversely, any polynomial-factor randomized approximation would yield a randomized polynomial-time factoring algorithm.  Thus, if \(\Factor\notin\FBPP\), then \(\CMO\in\BQFPTAS\setminus\RPOLYAPX\) (\Cref{thm:cmo}).  The same witness makes each of the \(\RAPX\subseteq\BQAPX\), \(\RPTAS\subseteq\BQPTAS\), and \(\RFPTAS\subseteq\BQFPTAS\) containments strict.

Discrete-Logarithm Fitting gives another conditional quantum--classical gap, this time at a one-half approximation threshold.  It scores a hypothesis by the fraction of labels it predicts correctly on an explicit sample.  Predicting the more common label gives a deterministic \(1/2\)-approximation, while quantum discrete logarithm finds a perfect fit.  Any fixed randomized improvement over \(1/2\) would solve discrete logarithm in a prime-order subgroup modulo a safe prime \(p=2q+1\), where \(q\) is also prime.  Under the assumption that this search problem is not in \(\FBPP\), \(\DLogFit\in\RAPX\cap\BQFPTAS\setminus\RPTAS\) (\Cref{thm:keyfit} in Appendix~\ref{app:keyfit}).

\Cref{prop:search-padding} shows how a similar one-half threshold can arise generically from a search task with efficiently verifiable witnesses.  It constructs a problem in NPO.  Every instance has an easy feasible dummy solution of score one, while a proposed witness of score two can be checked in polynomial time even if finding one is classically hard.  The dummy gives a deterministic \(1/2\)-approximation, and any fixed improvement over \(1/2\) recovers a witness.  If quantum polynomial time can find the witnesses but randomized classical polynomial time cannot, the constructed problem lies in \(\RAPX\cap\BQFPTAS\setminus\RPTAS\).  Unlike Discrete-Logarithm Fitting's objective based on sample accuracy, this construction has only two possible scores.

For bounded maximization scores, the decoded-mean transfer in \Cref{cor:collapse} turns an efficiently decoded expected score into a bounded-error classical-output guarantee.  Under \(\NP\nsubseteq\BQP\), the companion paper rules out uniformly efficient state preparation that recovers a fixed fraction of the optimal MaxCut improvement above the random-cut baseline on every graph~\cite[Theorem~1]{HadfieldApproxQOpt2026}.  Fixed-depth QAOA on unweighted 3-regular MaxCut yields a cut guarantee for that graph family after repetition (\Cref{sec:decoder}).  Computing the value \(\lambda(N)\) alone does not return a certified optimizer for \(\CMO\) (\Cref{sec:cmo}).

\Cref{fig:class-containments} places five example problems at four class positions, with Max-E3SAT and MaxCut sharing point A.  MaxCut lies in both constant-ratio classes and, if \(\NP\nsubseteq\BQP\), outside \(\BQPTAS\).  \Cref{tab:extended-witnesses} and Appendix~\ref{app:growing-proofs} give a conditional growing-factor hierarchy using classical witnesses.  Previously, Pirnay et al. gave conditional quantum advantages on constructed cryptographic instance families, while Szegedy outlined a broader gap-reduction route~\cite{Pirnay2024,Szegedy2022}.  For unrestricted MaxCut or MaxSAT, the central open problem is a quantum algorithm output ratio beyond what randomized classical algorithms can attain on the same problem.  For QAOA and compressed encodings, another task is to bound both the quality of the returned solutions and the cost of producing them.

We give the class definitions in Section~\ref{sec:classes}, the decoded-mean result in Section~\ref{sec:decoder}, and the separating constructions in Section~\ref{sec:arithmetic}.  Sections~\ref{sec:separation-implications} and~\ref{sec:scope} discuss implications of our results and connections to existing quantum output guarantees.
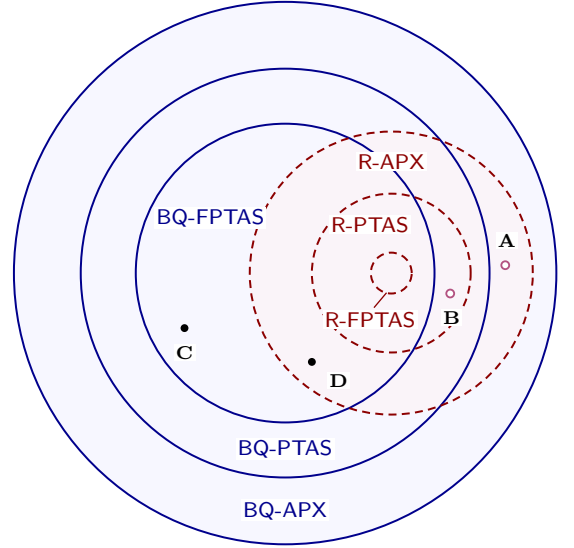
\begin{figure}[t]
\centering
\begin{tikzpicture}[x=1.04cm,y=1.04cm,
  quantum/.style={draw=blue!55!black, line width=.75pt},
  classical/.style={draw=red!55!black, densely dashed, line width=.75pt},
  name/.style={font=\footnotesize, fill=white, inner sep=.7pt}]
  \fill[blue!3] (0,0) circle[radius=3.45];
  \fill[red!5,fill opacity=.45] (1.35,0) circle[radius=1.8];
  \draw[quantum] (0,0) circle[radius=3.45];
  \draw[quantum] (0,0) circle[radius=2.6];
  \draw[quantum] (0,0) circle[radius=1.9];
  \draw[classical] (1.35,0) circle[radius=1.8];
  \draw[classical] (1.35,0) circle[radius=1.01];
  \draw[classical] (1.35,0) circle[radius=.26];
  \node[name,text=blue!55!black] at (0,-3.02) {\(\BQAPX\)};
  \node[name,text=blue!55!black] at (0,-2.24) {\(\BQPTAS\)};
  \node[name,text=blue!55!black] at (-.95,.70) {\(\BQFPTAS\)};
  \node[name,text=red!55!black] at (1.35,1.42) {\(\RAPX\)};
  \node[name,text=red!55!black] at (1.08,.62) {\(\RPTAS\)};
  \draw[red!55!black, line width=.5pt] (1.29,-.25)--(1.13,-.44);
  \node[name,text=red!55!black] at (1.08,-.59) {\(\RFPTAS\)};
  \draw[magenta!70!black,fill=white,line width=.7pt] (2.80,.10) circle[radius=.05];
  \node[name,font=\scriptsize\bfseries] at (2.82,.42) {A};
  \draw[magenta!70!black,fill=white,line width=.7pt] (2.10,-.26) circle[radius=.05];
  \node[name,font=\scriptsize\bfseries] at (2.12,-.56) {B};
  \fill[black] (-1.28,-.70) circle[radius=.05];
  \node[name,font=\scriptsize\bfseries] at (-1.28,-1.01) {C};
  \fill[black] (.34,-1.13) circle[radius=.05];
  \node[name,font=\scriptsize\bfseries] at (.68,-1.37) {D};
\end{tikzpicture}
\caption{Schematic class inclusions with four markers for five example problems.  Solid blue circles show quantum classes and dashed red circles randomized classical classes.  At each level, the deterministic class \(\APX\), \(\PTAS\), or \(\FPTAS\) is contained in its randomized and quantum counterparts.  The deterministic boundaries are not drawn separately.  Open markers \textbf{A} and \textbf{B} have deterministic approximation algorithms; filled markers \textbf{C} and \textbf{D} mark conditional quantum--classical ratio separations.  Under \(\NP\nsubseteq\BQP\), Max-E3SAT and MaxCut share marker \textbf{A} in \(\RAPX\setminus\BQPTAS\), while Planar Vertex Cover occupies marker \textbf{B} in \(\RPTAS\setminus\BQFPTAS\)~\cite{Hastad2001,GoemansWilliamson1995,Baker1994}.  Marker \textbf{C} (\(\CMO\)) lies in \(\BQFPTAS\setminus\RAPX\) if \(\Factor\notin\FBPP\), and is even outside \(\RPOLYAPX\) (\Cref{thm:cmo}).  Marker \textbf{D} is the sample-fitting threshold example \(\DLogFit\) from Appendix~\ref{app:keyfit}.  It lies in \(\BQFPTAS\setminus\RPTAS\) if safe-prime discrete logarithm is not in \(\FBPP\).}
\label{fig:class-containments}
  \end{figure}

\section{Classical-output quantum approximation classes}
\label{sec:classes}

A maximization problem in NP optimization (NPO), \(\Pi=(I,F,v)\), has polynomial-time recognizable instances \(x\in I\), polynomially bounded candidate strings, a nonempty feasible set \(F(x)\) recognized in deterministic polynomial time, and a nonnegative rational objective \(v(x,y)\) computable in polynomial time~\cite{PapadimitriouYannakakis1991,Ausiello1999}.  We assume
\(\OPT(x)=\max_{y\in F(x)}v(x,y)>0\).  The score of a feasible \(y\) is \(v(x,y)\), and its achieved ratio is \(v(x,y)/\OPT(x)\in[0,1]\).  For minimization, \(v\) is a strictly positive cost and \(\OPT\) its minimum, giving a cost ratio at least one.

\begin{definition}[Bounded-error output approximation]
\label{def:output}
A uniform randomized classical or quantum algorithm is an \(\alpha\)-approximation for a maximization problem, with fixed \(\alpha\in(0,1]\), if, on every input \(x\), it outputs \(y\in F(x)\) satisfying
\[
  v(x,y)\ge \alpha\OPT(x)
\]
with probability at least \(2/3\).  An invalid output is a failure.  For minimization, the corresponding condition is \(v(x,y)\le\rho\OPT(x)\), where \(\rho\ge1\).
\end{definition}

The \(2/3\) threshold is a conventional success probability, distinct from the approximation ratio.  For example, suppose a uniform polynomial-time procedure prepares a quantum state, and an efficient measurement and decoder return a solution meeting the claimed ratio with probability at least \(1/q(n)\) on every \(n\)-bit input, for a known polynomial \(q(n)\ge1\).  Repeating the entire procedure \(\lceil q(n)\ln 3\rceil\) times and keeping the best feasible output raises success to at least \(2/3\), without knowing \(\OPT\).  For schemes, the repetition count must respect the \(\PTAS\) or \(\FPTAS\) time bound; for an \(\FPTAS\), it must be polynomial jointly in \(n\) and \(1/\epsilon\).  Conversely, \(O(\log(2+n))\) repetitions reduce error from at most \(1/3\) to any fixed inverse polynomial in \(n\).  All steps and repetitions count toward running time; see \Cref{prop:conventions}.

All algorithms are uniformly specified.  A quantum algorithm may use intermediate measurements but no quantum advice.  Circuit generation, classical parameter search (including its circuit evaluations), preparation, measurement shots, decoding, verification, repetitions, and output selection all count toward its time.  For a checkable search relation \(S\), the shorthand \(S\in\FBPP\) means that a randomized polynomial-time algorithm returns a valid witness with probability at least \(2/3\) on every input of the stated domain, which may be a promise family.  It makes no claim about inputs outside that domain.  We use this convention for \(\Factor\) and \(\DLog\) below.

\begin{definition}[Approximation classes]
\label{def:classes}
\(\RAPX\) is the randomized analogue of classical \(\APX\).  It contains the NPO problems admitting a randomized polynomial-time \(\alpha\)-approximation for some constant \(\alpha>0\) in the maximization convention, or some constant \(\rho\ge1\) for minimization.  Let \(\epsilon\in\mathbb Q\cap(0,1)\) be supplied in binary, and let \(\ell(\epsilon)\) denote its encoding length.  A problem lies in \(\RPTAS\) if one uniform randomized algorithm \(\mathcal A\) receives \((x,\epsilon)\), succeeds with probability at least \(2/3\), and, for every fixed \(\epsilon\), has running time at most \(p_\epsilon(|x|)\) for some polynomial \(p_\epsilon\).  Its guarantee is \((1-\epsilon)\) for maximization and \((1+\epsilon)\) for minimization.  It lies in \(\RFPTAS\) if one polynomial bounds the running time in \(|x|\), \(1/\epsilon\), and \(\ell(\epsilon)\), uniformly for every \((x,\epsilon)\).  Replacing randomized classical computation by a uniform bounded-error quantum algorithm defines \(\BQAPX\), \(\BQPTAS\), and \(\BQFPTAS\) with the same quantifiers and guarantees.
\end{definition}

The deterministic classes obey \(\APX\subseteq\RAPX\), \(\PTAS\subseteq\RPTAS\), and \(\FPTAS\subseteq\RFPTAS\).

\paragraph{Accuracy encoding.}
The term \(\ell(\epsilon)\) charges the time to read a binary rational.  For instance, dyadic accuracies \(\epsilon=2^{-k}\) have \(\ell(\epsilon)=O(1+\log(1/\epsilon))\), whereas an arbitrary rational near one can have a much longer encoding.  Our placements use fixed accuracies or polynomial-length encodings and are insensitive to this choice.  The output convention also applies to other accuracy bounds, such as an \(\mathsf{EPTAS}\) with time \(f(1/\epsilon)\operatorname{poly}(|x|)\) and polynomial degree independent of \(\epsilon\)~\cite{CesatiTrevisan1997}.  Each quantum analogue retains classical feasibility and score tests and charges the full algorithm.  An exact bounded-error quantum optimizer may ignore \(\epsilon\) and belongs to all three quantum classes.  The immediate containments are
\[
\begin{gathered}
 \RFPTAS\subseteq\RPTAS\subseteq\RAPX,\\
 \BQFPTAS\subseteq\BQPTAS\subseteq\BQAPX,\\
 \RAPX\subseteq\BQAPX,\\
 \RPTAS\subseteq\BQPTAS,\\
 \RFPTAS\subseteq\BQFPTAS.
\end{gathered}
\]

Membership transfers backward under an explicit ratio-preserving solution reduction.  For maximization problems \(\Pi\) and \(\Gamma\), this means polynomial-time maps \(f\) on instances and \(g\) on feasible outputs, with \(g(x,z)\in F_\Pi(x)\) for every \(z\in F_\Gamma(f(x))\), such that
\[
 \frac{v_\Pi(x,g(x,z))}{\OPT_\Pi(x)}\ge
 \frac{v_\Gamma(f(x),z)}{\OPT_\Gamma(f(x))}.
\]
Composing these maps with an approximation algorithm for \(\Gamma\) preserves its ratio and success probability for \(\Pi\).  For minimization, the analogous condition reverses the displayed inequality.  Polynomial input and output overhead preserve each of the corresponding \(\APX\), \(\PTAS\), and \(\FPTAS\) time bounds, including their randomized and quantum versions.  Under \(\Factor\notin\FBPP\), no classical \(\APX\) problem can be \(\BQAPX\)-complete under these reductions: \(\CMO\in\BQAPX\setminus\APX\) by \Cref{thm:cmo}, and a reduction to a classical \(\APX\) problem would give \(\CMO\) a classical constant-factor algorithm.

The following hierarchy is a direct adaptation of standard classical witnesses.  Its role is to check that the definitions retain the familiar distinctions among \(\APX\), \(\PTAS\), and \(\FPTAS\).

\begin{proposition}[Conditional quantum hierarchy]
\label{prop:hierarchy}
If \(\NP\nsubseteq\BQP\), then
\[
  \BQFPTAS\subsetneq\BQPTAS\subsetneq\BQAPX.
\]
\end{proposition}

\noindent Appendix~\ref{app:hierarchy} gives the gap and integrality arguments using Max-E3SAT and Planar Vertex Cover.  MaxCut provides another witness for \(\BQPTAS\subsetneq\BQAPX\).  On unweighted graphs with an edge, a deterministic \(1/2\)-approximation gives \(\mathsf{MaxCut}\in\APX\subseteq\RAPX\subseteq\BQAPX\), while its NP-hard approximation gap excludes \(\BQPTAS\) if \(\NP\nsubseteq\BQP\)~\cite{Hastad2001,GoemansWilliamson1995}.  The companion paper derives the same \(\BQPTAS\) exclusion from decoded means by converting any fixed ratio above MaxCut's hardness threshold into a bounded-error output guarantee~\cite[Prop.~1]{HadfieldApproxQOpt2026}.

A restriction can change the ratio and class placement.  On bipartite graphs with an edge, an efficiently found bipartition cuts every edge, so the restricted MaxCut problem is exactly solvable and lies in \(\FPTAS\).

\begin{proposition}[Uniform conventions and amplification]
\label{prop:conventions}
In the definitions above, replacing success probability \(2/3\) by any other fixed constant in \((0,1)\) leaves each class unchanged.  For schemes, one can demand success probability at least \(1-\xi\), for a specified \(\xi\in(0,1)\), with an additional \(O(\log(1/\xi))\) factor in running time, provided its encoding is charged.  Fixing \(\epsilon\) in a uniform \(\PTAS\) yields an algorithm family.  The converse requires an effective uniform generator for that family with the stipulated running-time bound.
\end{proposition}

\begin{proof}
Run independently, reject invalid outputs, and retain the best feasible maximization candidate (or the minimization candidate with the lowest cost).  If a run succeeds with probability at least a fixed \(c>0\), \(k\) trials succeed with probability at least \(1-(1-c)^k\), without knowing \(\OPT\).  Taking \(k=O(\log(1/\xi))\) gives the stated bound.  With fixed \(k\), the scheme runtimes are unchanged up to constants.  For \(\FPTAS\), a requested \(\xi\) such as \(\epsilon^a\) adds only polynomial overhead in \(1/\epsilon\).  The statement about families follows by specializing the uniform algorithm.  Even if each fixed accuracy has a polynomial-time algorithm, the map from accuracy to algorithm need not be computable, so the family alone supplies no uniform \(\PTAS\).
\end{proof}

The same output convention admits factors that grow with the full input length \(N=|x|\).  \Cref{tab:extended-witnesses} summarizes further placements.  In the first three rows, LOG, POLY, and EXP describe logarithmic, polynomial, and exponential approximation factors in \(N\), while the algorithms still run in polynomial time.  Appendix~\ref{app:growing-proofs} gives the precise class definitions and proofs.

\begin{table}[b]
\centering
\footnotesize
\setlength{\tabcolsep}{4pt}
\begin{tabular}{@{}p{0.28\linewidth}p{0.30\linewidth}p{0.34\linewidth}@{}}
\toprule
Problem & Classical status & Quantum status \\
\midrule
Set Cover~\cite{Chvatal1979,Feige1998,DinurSteurer2014} & Logarithmic factor (\(\RLOGAPX\)) & Outside \(\BQAPX\) \\
Max Clique~\cite{Zuckerman2007} & Polynomial factor (\(\RPOLYAPX\)) & Outside \(\BQLOGAPX\) \\
Positive-Weight TSP~\cite{SahniGonzalez1976} & Exponential factor (\(\REXPAPX\)) & Outside \(\BQPOLYAPX\) \\
\(\CMO\) (\Cref{thm:cmo}) & In \(\REXPAPX\); outside \(\RPOLYAPX\) if \(\Factor\notin\FBPP\) & In \(\BQFPTAS\) \\
Bin Packing~\cite{FernandezDeLaVegaLueker1981} & \(\RAPX\cap\RAPTAS\) & Outside \(\BQPTAS\) \\
Rectilinear TSP~\cite{Arora1998,GareyGrahamJohnson1976} & \(\PTAS\) (\(\RPTAS\)) & Outside \(\BQFPTAS\) \\
0--1 Knapsack~\cite{IbarraKim1975} & \(\FPTAS\) (\(\RFPTAS\)) & In \(\BQFPTAS\) \\
\bottomrule
\end{tabular}
\caption{Further placements in the classical-output model.  TSP denotes the Traveling Salesperson Problem.  Its positive-weight row uses binary-encoded integer weights, and the rectilinear row uses integer-grid points.  The randomized asymptotic scheme \(\RAPTAS\) is defined in \Cref{prop:bin-packing}.  The first three rows establish \Cref{prop:growing-hierarchy}.  The CMO exponential-factor upper bound is unconditional; its classical lower bound assumes \(\Factor\notin\FBPP\).  Other quantum exclusions assume \(\NP\nsubseteq\BQP\).  The 0--1 Knapsack row is a membership example.}
\label{tab:extended-witnesses}
\end{table}

The Set Cover, Max Clique, both TSP, and Bin Packing exclusions follow from \(\NP\)-hardness.  The CMO and DLog-Fit witnesses use factoring and discrete logarithm, respectively.

\section{Decoded mean and classical output}
\label{sec:decoder}

Many quantum optimization results bound an expectation, while the classes above require a good classical solution with bounded error.  Here we show how an expected score after efficient decoding yields such an output guarantee, and conversely how an output guarantee yields a high decoded mean.  This connects expectation analyses to class membership when the required uniformity and runtime bounds hold.  We work through a fixed-depth MaxCut example.  For methods with nontrivial decoding, a state's energy need not equal the expected score of the decoded solution~\cite{Jordan2025,HadfieldApproxQOpt2026}.

Fix a maximization instance \(x\) with scores \(0\le v(x,y)\le V^*=\OPT(x)\).  An efficient decoder has outcomes \(F(x)\cup\{\bot\}\), where \(\bot\) denotes failure and has score \(v(x,\bot)=0\).  If its measurement is a positive operator-valued measure (POVM) \(\{M_y\}_{y\in F(x)}\cup\{M_\bot\}\), whose positive effects sum to the identity, its score observable is
\[
  G=\sum_{y\in F(x)}v(x,y)M_y,
\]
so, for a fixed decoder and prepared density operator \(\rho\), \(\Tr(G\rho)=\mathbb E[v(x,Y)\mid\rho,\text{decoder}]\).  If preparation or decoding uses classical randomness, average this identity over those choices.

\begin{definition}[Decoded-mean approximation]
A decoded-mean \(\alpha\)-approximation, for \(\alpha\in(0,1]\), is a uniform polynomial-time state-preparation and decoding algorithm satisfying \(\mathbb E[v(x,Y)]\ge\alpha\OPT(x)\) on every maximization instance.  Decoded-mean \(\PTAS\) and \(\FPTAS\) variants require mean at least \((1-\epsilon)\OPT(x)\) with the time bounds in the preceding definition.  The complete preparation and decoder, including their choice of parameters, are uniform in \((x,\epsilon)\) and may be rerun independently with fresh quantum and classical randomness.  Failure has score zero.  Every other decoder outcome is a polynomial-length feasible solution whose score can be computed efficiently.
\end{definition}

The elementary bounded-score inequality below was used for MaxCut in the companion paper's proofs for ordinary approximation ratios and gain barriers~\cite[Prop.~1 and Eq.~(14)]{HadfieldApproxQOpt2026}.  We state it for arbitrary decoded scores and use it in the three approximation-class regimes.

\begin{lemma}[Bounded-score mean-to-tail bound]
\label{lem:mean-tail}
Let \(X=v(Y)/V^*\in[0,1]\) and suppose \(\mathbb E X\ge\mu\).  For every \(0\le\tau<\mu\le1\),
\begin{equation}
  \Pr\!\left[X>\tau\right]\ge\frac{\mu-\tau}{1-\tau}.
  \label{eq:mean-tail}
\end{equation}
The bound is tight for every fixed \((\mu,\tau)\).
\end{lemma}

\begin{proof}
Writing \(p=\Pr[X>\tau]\), boundedness gives
\(
  \mathbb E X\le \tau(1-p)+p=\tau+(1-\tau)p.
\)
Rearranging proves the bound.  A two-point distribution supported at \(\tau\) and \(1\) attains equality.
\end{proof}

\begin{corollary}[Equivalence of decoded-mean and output classes]
\label{cor:collapse}
For maximization problems satisfying the NPO assumptions of \Cref{sec:classes}, suppose the entire preparation--decoder procedure is one uniform algorithm in \((x,\epsilon)\) with the stated class time bound and may be independently rerun.  Its decoded mean is \(\mathbb E[v(x,Y)]\) for the classical candidate \(Y\) returned by that same decoder, with \(Y\in F(x)\cup\{\bot\}\) and \(v(x,\bot)=0\).  Then decoded-mean and bounded-error classical-output versions of \(\BQAPX\), \(\BQPTAS\), and \(\BQFPTAS\) coincide up to constant or \(\epsilon\)-rescaling.  The output algorithm may use the composite decoder that repeats and selects the best valid candidate.
\end{corollary}

\begin{proof}
For the mean-to-output direction, apply \Cref{lem:mean-tail}.  A decoded-mean \(\PTAS\) run at \(\epsilon/2\) has \(\mu\ge1-\epsilon/2\).  Taking \(\tau=1-\epsilon\) gives success probability at least \(1/2\), and two independent runs raise it to at least \(3/4\).  The same rescaling preserves \(\FPTAS\) time.  At constant ratio, mean \(\alpha\) gives output ratio \(\alpha/2\) with probability at least \(\alpha/(2-\alpha)\), followed by constant amplification.

Conversely, run a bounded-error output scheme at parameter \(\epsilon/2\) independently
\(
 k=\lceil\log_3(2/\epsilon)\rceil
\)
times, verify and score its candidates, and retain the best valid one.  If every run is invalid, output \(\bot\).  The all-failure probability is at most \(\epsilon/2\), so the expected normalized score is at least
\[
  (1-\epsilon/2)^2\ge1-\epsilon.
\]
For \(\APX\), a bounded-error \(\alpha\)-approximation has decoded mean at least \(2\alpha/3\), and constant repetition can make the success contribution arbitrarily close to \(\alpha\).
\end{proof}

The companion paper uses the forward mean-to-output implication for MaxCut~\cite[Prop.~1]{HadfieldApproxQOpt2026}.  The converse above completes the equivalence for the three uniform classes.  Section~\ref{sec:scope} applies the same score and resource accounting to recent DQI output results.

The corollary concerns the expectation of the \emph{decoder-induced} score observable \(G\).  A high expectation for a different surrogate Hamiltonian \(H\) requires a separate argument relating \(H\) to \(G\) before the corollary applies (Table~\ref{tab:interfaces}).

The transfer above is for maximization.  For an always-feasible minimization decoder with \(v(Y)\ge\OPT>0\), an expected cost at most \((1+\epsilon')\OPT\) gives
\(\Pr[v(Y)>(1+\epsilon)\OPT]\le\epsilon'/\epsilon\)
by Markov's inequality applied to the nonnegative excess cost.  Taking \(\epsilon'=\epsilon/3\) gives success probability at least \(2/3\).  The convention that assigns score zero on failure does not apply to minimization.

\noindent\emph{Example. QAOA on 3-regular MaxCut.}\enspace
Computational-basis measurement of ordinary MaxCut QAOA returns a cut with score \(C(Y)\).  Let \(H_C\) be the diagonal cost Hamiltonian whose computational-basis eigenvalues are these cut scores.  Its score observable is therefore \(G=H_C\).  Farhi, Goldstone, and Gutmann's analysis of unweighted 3-regular graphs gives a one-layer (\(p=1\)) choice with expected cut value at least \(0.6924\,\OPT\), using efficient fixed-depth preparation and preprocessing~\cite[Sec.~V, Eq.~(35)]{Farhi2014}.  Taking \(\tau=2/3\) in \Cref{lem:mean-tail} yields
\[
  \Pr\!\left[C(Y)>(2/3)\OPT\right]
  \ge 3\bigl(0.6924-2/3\bigr)=0.0772.
\]
The best of 14 independent measured cuts therefore exceeds \((2/3)\OPT\) with probability at least \(1-(1-0.0772)^{14}>2/3\), giving \(\BQAPX\) membership for this graph family.  Thirteen runs would not reach success probability \(2/3\) from this mean-only bound.  Fixed-depth QAOA on bounded-degree graphs also has cut-score variance \(O(|E|)\) by locality~\cite[Sec.~III, Eq.~(30)]{Farhi2014}.  Since \(\OPT\ge |E|/2\), Chebyshev's inequality makes the one-shot failure probability at the same \(2/3\) threshold \(O(1/|E|)\) for \(p=1\).  Thus one run suffices on sufficiently large graphs; the 14-run bound holds uniformly at every size.  The general Goemans--Williamson algorithm gives a stronger classical guarantee~\cite{GoemansWilliamson1995}, and a semidefinite-programming-based algorithm guarantees at least \(0.9326\) on unweighted graphs of maximum degree three~\cite{HalperinLivnatZwick2004}.  Full MaxCut belongs to \(\BQAPX\) through classical approximation, while this QAOA calculation concerns the 3-regular family.

\begin{table}[t]
\centering
\footnotesize
\setlength{\tabcolsep}{4pt}
\begin{tabular}{@{}p{0.27\linewidth}p{0.31\linewidth}p{0.34\linewidth}@{}}
\toprule
\raggedright Reported result\par & \raggedright Guarantee\par & \raggedright Remaining step\par \tabularnewline
\midrule
\raggedright Value estimate\par & \raggedright $\widehat V\approx\OPT$\par & \raggedright Find a feasible solution and prove its ratio\par \tabularnewline
\raggedright Surrogate energy\par & \raggedright High $\Tr(H\rho)$\par & \raggedright Relate $H$ to decoded score $G$\par \tabularnewline
\raggedright Decoded mean\par & \raggedright High $\mathbb E[v(x,Y)]$\par & \raggedright Bound scores and amplify success\par \tabularnewline
\raggedright Classical output\par & \raggedright Feasible $Y$ meets the ratio with probability $\ge 2/3$\par & \raggedright Count total running time\par \tabularnewline
\bottomrule
\end{tabular}
\caption{Four interfaces for maximization. Here $\widehat V$ estimates $\OPT$, $H$ is a surrogate Hamiltonian, $\rho$ a prepared state, and $Y$ a decoded candidate. The last row gives the target guarantee, subject to the full runtime bound.}
\label{tab:interfaces}
\end{table}

Fixed-angle guarantees at other depths can be converted in the same way.  Wurtz and Love proved that fixed-angle \(p=2\) QAOA has mean ratio at least \(0.7559\) on every unweighted 3-regular MaxCut graph~\cite{WurtzLove2021}.  This bound likewise gives a bounded-error output guarantee.  Appendix~\ref{app:qaoa-repetitions} gives the repetition count for any uniformly implementable circuit with a stated worst-case mean ratio.  Higher-depth analyses on large-girth graphs (graphs without short cycles) give further mean bounds for those graph families~\cite{BassoFarhiMarwahaVillalongaZhou2022,FarhiGutmannRanardVillalonga2025}.  Farhi et al. report an expected cut fraction of approximately \(0.8971\) at fixed \(p=17\) on unweighted 3-regular graphs of girth at least \(36\), using angles precomputed for this family~\cite[Sec.~5.1, Table~1]{FarhiGutmannRanardVillalonga2025}.  Since \(\OPT\le |E|\), a cut-fraction bound also lower-bounds the expected approximation ratio, but this result does not improve the \(0.9326\) classical ratio guarantee quoted above.

\section{Search padding and Certified Maximum Order}
\label{sec:arithmetic}

This section develops conditional gaps between quantum and randomized classical approximation.  \Cref{prop:search-padding} turns a classically checkable search relation into an NPO problem with a one-half approximation threshold.  Certified Maximum Order has graded element-order scores, so even a polynomial-factor randomized approximation would yield a factoring algorithm (\Cref{thm:cmo}).

\subsection{The generic search threshold}

\begin{proposition}[Search padding]
\label{prop:search-padding}
Let \(I\) be a polynomial-time recognizable set and \(R(x,y)\) a polynomial-time decidable, polynomially balanced relation with at least one witness \(y\) for every \(x\in I\).  There is an NPO maximization problem \(\Pi_R\) with a deterministic \(1/2\)-approximation such that
\[
  \Pi_R\in\RPTAS\quad\Longleftrightarrow\quad R\in\FBPP.
\]
Indeed, a randomized polynomial-time \(\alpha\)-approximation for any fixed \(\alpha>1/2\) gives a solver for \(R\), while a randomized solver for \(R\) gives an exact optimizer for \(\Pi_R\).  A bounded-error quantum polynomial-time solver for \(R\) likewise gives an exact quantum optimizer.  Thus, if \(R\) has a quantum solver but no randomized solver on \(I\), then \(\Pi_R\in\RAPX\cap\BQFPTAS\setminus\RPTAS\).
\end{proposition}

\begin{proof}
For input \(x\in I\), allow a dummy solution \(\mathsf{d}\) of score one and tagged witnesses \((\mathsf{w},y)\) of score two precisely when \(R(x,y)\) holds.  The optimum is two and the dummy gives ratio \(1/2\).  Verify the output of a quantum or randomized search algorithm and fall back to \(\mathsf{d}\) on failure.  Either search algorithm is an exact optimizer on its successful event and ignores \(\epsilon\).  Conversely, an \(\alpha>1/2\) approximation output must be a valid tagged witness on its successful event.  A randomized \(\PTAS\) called at \(\epsilon=1/4\) therefore solves \(R\) with probability at least \(2/3\).
\end{proof}

\Needspace{7\baselineskip}
\DLogFit{} adapts the cryptographic learning-to-distinguishing idea of Kearns and Valiant and later discrete-logarithm learning work to an explicit finite-sample optimization problem~\cite{KearnsValiant1994,LiuArunachalamTemme2021}.  Its hypotheses, feasibility, and scores are efficiently checkable on the input list.  Appendix~\ref{app:keyfit} defines the problem and proves \Cref{thm:keyfit} by a direct reduction.  A deterministic algorithm attains ratio \(1/2\), a quantum algorithm finds an optimum, and any fixed randomized classical ratio above \(1/2\) would solve safe-prime discrete logarithm.

\subsection{Certified Maximum Order}
\label{sec:cmo}

The ordinary period-finding problem fixes a modular multiplier and asks for its order.  We instead optimize over the multiplier.  A reported order divides the maximum possible order, leaving only polynomially many missing multipliers for a polynomial-factor approximation.  The certificate makes that reported order and its score efficiently checkable.

For any positive integer \(m\), let \(U_m=(\mathbb Z/m\mathbb Z)^\times\) be the multiplicative group of residues coprime to \(m\).  For \(N\ge3\), the multiplicative order \(\ord_N(a)\) is the least positive \(r\) with \(a^r\equiv1\pmod N\).  The Carmichael function \(\lambda(N)\) is the least positive exponent that works for \emph{every} \(a\in U_N\), equivalently the least common multiple of their orders.  Since \(U_N\) is finite abelian, some element has order \(\lambda(N)\)~\cite[Sec.~II and Theorem~6]{GarciaEscartin2021}.

\begin{definition}[Certified Maximum Order problem]
An instance of \(\CMO\) is an integer \(N\ge3\) in binary.  A feasible solution is a tuple
\[
  (a,r,\{(q_j,e_j)\}_{j=1}^k)
\]
such that \(1\le a<N\), \(\gcd(a,N)=1\), \(1\le r\le N\), the listed distinct primes and positive exponents satisfy \(r=\prod_jq_j^{e_j}\), and
\[
\begin{gathered}
  a^r\equiv1\pmod N,\\
  a^{r/q_j}\not\equiv1\pmod N\quad\text{for every }j.
\end{gathered}
\]
The objective value is \(r\).
\end{definition}

The list is empty when \(r=1\), and \((1,1,\varnothing)\) is a feasible fallback.  Since \(2^k\le\prod_{j=1}^kq_j^{e_j}=r\le N\), one has \(k\le\log_2N\).  Each \(q_j\le N\) and \(e_j\le\log_2N\) as well, so solutions are polynomially balanced.  The verifier first rejects malformed encodings and violations of these bounds.  Only then does it use exact-integer multiplication, primality testing, gcds, and modular exponentiation to verify feasibility in polynomial time~\cite{AKS2004}.  These conditions certify \(r=\ord_N(a)\).  If the true order were a proper divisor of \(r\), deleting some prime factor from \(r\) would still give an exponent.  Checking only \(a^r\equiv1\pmod N\) would instead admit \((a,r)=(1,N)\) on every input and trivialize the objective.  Consequently,
\[
  \OPT_{\CMO}(N)=\lambda(N).
\]
For \(n\)-bit \(N\), the fallback has ratio \(1/\lambda(N)>2^{-n}\).  Thus \(\CMO\in\REXPAPX\) unconditionally, while its finer classical placement below depends on factoring.

Garc{\'i}a-Escart{\'i}n gave a bounded-error quantum algorithm for the exact value \(\lambda(N)\), taking the least common multiple of orders of sampled units without first factoring \(N\)~\cite{GarciaEscartin2021}.  This addresses the optimum-value interface in \Cref{tab:interfaces}.  For \(\CMO\), a classical verifier also needs an element attaining the reported order and a prime factorization of that order.  Our construction in Appendix~\ref{app:cmo-proofs} factors \(N\) and the component group orders, builds a maximum-order element by the Chinese remainder theorem, and accounts for worst-case bounded error.  The converse applies even to polynomial-factor classical approximations.

\Needspace{8\baselineskip}
Let \(\Factor\) denote the function problem of returning a nontrivial factor of a composite integer.

\begin{lemma}[Factoring from an exponent]
\label{lem:exponent-factor}
Given an odd integer \(N\) with at least two distinct prime factors and a positive multiple \(M\) of \(\lambda(N)\), a randomized polynomial-time algorithm returns a nontrivial factor of \(N\) with constant success probability~\cite{Miller1976,Shor1997}.
\end{lemma}

\noindent See Appendix~\ref{app:cmo-proofs} for the uniform component analysis and proof.

\begin{theorem}[Certified Maximum Order versus factoring]
\label{thm:cmo}
\(\CMO\in\BQFPTAS\).  In fact, a bounded-error quantum polynomial-time algorithm returns an exact optimum.  Conversely, any randomized polynomial-time approximation for \(\CMO\) with factor at most \(Cn^c\), where \(n\) is the binary input length and \(C\ge1\) and \(c\ge0\) are fixed constants, yields a randomized polynomial-time factoring algorithm.  Hence \(\CMO\in\RPOLYAPX\) if and only if \(\Factor\in\FBPP\).  If \(\Factor\notin\FBPP\), then
\[
  \CMO\in\BQFPTAS\setminus\RPOLYAPX,
\]
so \(\RAPX\subsetneq\BQAPX\), \(\RPTAS\subsetneq\BQPTAS\), and \(\RFPTAS\subsetneq\BQFPTAS\).  The matched logarithmic- and polynomial-factor quantum classes are likewise strictly larger than their randomized classical counterparts.
\end{theorem}

\noindent The full proof is in Appendix~\ref{app:cmo-proofs}.  A certified approximate order \(r\) divides \(\lambda(N)\).  For a \(Cn^c\)-factor guarantee, the missing multiplier \(\lambda(N)/r\) is an integer at most \(Cn^c\).  The reduction tries every multiplier up to this bound before applying \Cref{lem:exponent-factor}.  Conversely, if factoring is in \(\FBPP\), the constructive upper-bound algorithm in Appendix~\ref{app:cmo-proofs} uses randomized classical factoring and generator sampling to find an exact optimum in polynomial time.  The quantum construction and full success accounting are given there as well.  The growing-factor classes are defined in \Cref{def:growing-factors}.

For example, \(\lambda(35)=12\), and \(a=2\) has order \(12\).  Thus \((2,12,\{(2,2),(3,1)\})\) is an optimal certificate.  Given a certified order \(r=6\) from a one-half approximation, the reduction need only test multipliers \(k\in\{1,2\}\).  The second reaches exponent \(12\), and the square-root step reveals a factor.

More generally, suppose a polynomial-time classical algorithm has a factor guarantee \(\rho(n)\le B(n)\), where \(B(n)\) is an efficiently computable integer with \(\log B(n)=\operatorname{poly}(n)\).  Enumerating \(k\le B(n)\) in the same reduction gives a randomized factoring algorithm with time \(B(n)\operatorname{poly}(n+\log B(n))\).  For instance, a \(2^{\lceil n^{1/4}\rceil}\)-factor polynomial-time approximation would imply a \(2^{O(n^{1/4})}\)-time randomized factoring algorithm.  This is a quantitative consequence of the reduction, not an exclusion under the assumption \(\Factor\notin\FBPP\).

\begin{remark}[Exact-search origin]
\(\CMO\) and the appendix example \(\DLogFit\) both have bounded-error quantum polynomial-time algorithms that return an optimum.  Their cross-model separations arise from exact quantum search.  The polynomial-factor converse for \(\CMO\) uses its graded order objective and does not follow from search padding.
\end{remark}

\section{Interpreting same-problem separations}
\label{sec:separation-implications}

A quantum--classical comparison must fix the instance family, feasible outputs, score, and success criterion.  The Max-E3SAT and Planar Vertex Cover witnesses of \Cref{prop:hierarchy} establish strict quantum-class inclusions, while their approximation algorithms are already deterministic.

For \(\CMO\), a quantum algorithm finds an optimum exactly, while the factoring assumption excludes even polynomial-factor randomized approximation.  \(\DLogFit\) also has an exact quantum optimizer but admits a deterministic one-half approximation; its classical barrier excludes a \(\RPTAS\).  A mixed placement such as \(\BQPTAS\setminus\RFPTAS\) says less on its own.  If the problem is also in \(\RPTAS\), both models achieve every fixed accuracy, and \(\BQPTAS\) membership alone does not establish a better quantum dependence on \(\epsilon\).

A \(\BQAPX\setminus\RAPX\) witness would require randomized classical hardness against every fixed nonzero ratio, so it cannot be MaxCut or MaxSAT, which already lie in \(\APX\).  For either familiar problem, the question is instead whether a uniform quantum algorithm can guarantee a ratio \(\alpha_{\mathrm{qu}}\) above a classical hardness threshold \(\alpha_{\mathrm{cl}}\).  The problem could remain in both \(\RAPX\) and \(\BQAPX\); the separation would concern achievable ratios rather than class membership.

\section{Quantum output guarantees and their access models}
\label{sec:scope}

\subsection{Decoded quantum interferometry}

Decoded quantum interferometry (DQI) returns classically scoreable fits for Optimal Polynomial Intersection (OPI).  Jordan et al. obtain a higher expected satisfied fraction than known polynomial-time classical methods in their regime~\cite{Jordan2025}.  Subsequent analyses address circuit costs, worst-case bounds, and variants of the oracle access~\cite{Khattar2025,Rosmanis2026,SunWootters2026,HorinagaYamakawa2026,Jo2026}.

For an explicitly represented OPI family \(\mathcal F\) with efficient recognition and scoring, a uniform polynomial-time quantum algorithm returning fits with decoded mean at least \(\mu\OPT(x)\) on every \(x\in\mathcal F\) gives a bounded-error \(\tau\)-approximation for every fixed \(\tau<\mu\) by \Cref{lem:mean-tail} and independent repetition.  A ratio gap in this time model also needs randomized classical hardness at the same \(\tau\), beyond comparison with the best known scalar-OPI algorithm.  An \(\NP\)-hard gap for general max-LINSAT does not transfer to structured OPI~\cite{KramerSchubertEisert2026}.

For MaxCut, Parekh showed that the graphs behind an earlier asymptotic DQI guarantee are classically solvable in polynomial time~\cite{ParekhDQIMaxCut2025}.  More recently, Apte et al. obtained nontrivial DQI guarantees with imperfect decoding on specified regular MaxCut and Max-\(k\)-Cut families, while ruling out quantum advantage in the regimes they analyze~\cite{ApteCycleCodes2026}.  Table~\ref{tab:prior} distinguishes these results by output, classical comparator, and input access.

\begin{table*}[t]
\centering\footnotesize\setlength{\tabcolsep}{3.5pt}
\begin{tabular}{@{}p{0.19\linewidth}p{0.27\linewidth}p{0.23\linewidth}p{0.23\linewidth}@{}}
\toprule
Work & Instance and output & Classical comparator & Approximation statement \\
\midrule
Yamakami~\cite{Yamakami2002} & Quantum-index acceptance value & No randomized NPO comparator & Different output object \\
Pirnay et al.~\cite{Pirnay2024} & Classical solutions on constructed formula-coloring and integer-programming instance families & RSA-inversion-based classical hardness & Polynomial-factor classical barriers and efficient quantum approximation on those families \\
Szegedy~\cite{Szegedy2022}; Remmerswaal~\cite{Remmerswaal2026} & Classical solutions on selected NPO instances; paired formula--factoring gap-3SAT assignments & Efficient classical algorithms & PCP-based outline, then detailed fixed-ratio construction \\
Hadfield~\cite{HadfieldApproxQOpt2026} & Decoded MaxCut cuts from uniform state preparation & NP-hard MaxCut-Gain gap & Fixed positive all-input gain would imply \(\NP\subseteq\BQP\) \\
Jordan et al.~\cite{Jordan2025} (DQI on OPI) & Classical polynomial fits on a specified problem family & Best known polynomial-time classical algorithms & Better expected satisfied fraction in a stated regime (no proved class exclusion) \\
Kramer et al.~\cite{KramerGilFusterJonesEisertSchreiber2026} & Classical fits for folded OPI with membership oracles & Sampled-oracle classical query lower bound & Every-oracle DQI output above any fixed \(\tau\in(\alpha_{\mathrm C}(R),\alpha_{\mathrm Q}(R))\) \\
Wu and Yang~\cite{WuYang2026} & Polynomial fits for streamed OPI and Hermite OPI & Unconditional multipass classical space lower bound & Polynomial space gap for OPI; exponential for Hermite OPI \\
Here & Classical solutions to \(\CMO\), with \(\DLogFit\) in Appendix~\ref{app:keyfit} & \(\RPOLYAPX,\RPTAS\) & Polynomial-factor CMO barrier and appendix \(1/2\) threshold \\
\bottomrule
\end{tabular}
\caption{Classical-output results distinguished by instance family, access model, comparator, and resource bound.  In the folded-OPI row, \(\alpha_{\mathrm C}(R)\) and \(\alpha_{\mathrm Q}(R)\) denote its classical and DQI score thresholds at code rate \(R\), as specified below.  The oracle and streaming rows do not assert separation of the explicit-input time classes defined here.}
\label{tab:prior}
\end{table*}

\Needspace{7\baselineskip}
In recent work, Kramer et al. prove a membership-query output gap for balanced folded OPI~\cite[Eq.~(1.1) and Theorem~1]{KramerGilFusterJonesEisertSchreiber2026}.  Let \(R=\lim_{M\to\infty}K/(hM)\) be the folded Reed--Solomon rate, with \(K\) message coordinates and \(M\) blocks of \(h\) scalar evaluations.  Following their notation, the sampled-oracle classical threshold is \(\alpha_{\mathrm C}(R)=(1+R)/2\), while DQI's decoded mean approaches \(\alpha_{\mathrm Q}(R)=(1+\sqrt{R(2-R)})/2\) on every promised oracle.  Repetition gives verified outputs above each fixed \(\tau\in(\alpha_{\mathrm C}(R),\alpha_{\mathrm Q}(R))\) by \Cref{eq:mean-tail}; exceeding \(\alpha_{\mathrm C}(R)\) by a fixed amount with constant probability over sampled oracles requires super-polynomially many classical membership queries, even with unlimited between-query computation.  Since \(\OPT\le1\), these score thresholds also give approximation ratios.  A modified coherent procedure reaches any fixed score threshold below about \(0.958\) at \(R=0.3\) on typical sampled oracles~\cite[Theorem~2]{KramerGilFusterJonesEisertSchreiber2026}.  Two other recent papers establish streaming-space gaps: Wu and Yang give unconditional polynomial-space gaps for OPI and exponential-space gaps for Hermite OPI against multipass classical algorithms~\cite[Theorems~1.1 and 1.2]{WuYang2026}, while Bouland et al. give a one-pass quantum OPI algorithm for every stream order with a classical lower bound on random-order sampled inputs~\cite[Theorems~27 and 29]{BoulandDingJain2026}.  Their space bounds concern a different resource from the explicit-input time classes defined here.

\subsection{Quantum optimizers and encoded readout}

The output classes apply to more than QAOA.  Finite-time annealing and adiabatic optimization~\cite{KadowakiNishimori1998,FarhiAdiabatic2000}, dissipative and measurement-driven state preparation~\cite{CubittDQE2023,StollenwerkHadfield2025}, and recursive hybrids such as iterative quantum optimization and quantum-informed recursive optimization~\cite{BravyiKlieschKoenigTang2020,Dupont2023Greedy,Brady2024,FinzgarQIRO2024} are candidates for these classes.  Their final measured, decoded, or repaired outputs must have a proved ratio, with control choice and feedback included in the time bound~\cite[Sec.~IV]{HadfieldApproxQOpt2026}.

QAOA and constraint-preserving alternating operator ansatz variants return classical candidates~\cite{Abbas2024,Farhi2014,Hadfield2019,MarshWang2019}.  For ordinary MaxCut QAOA, computational-basis measurement has the cost Hamiltonian as its score observable.  A uniform worst-case mean bound therefore gives a bounded-error output ratio by \Cref{cor:collapse}, as illustrated in \Cref{sec:decoder}.  For constrained variants, the decoder must return a feasible candidate or declare failure; efficient repair is one way to do so.  Onah and Michielsen analyze high-probability exact output under inverse-polynomial ideal weight on optima, independent sampling, and efficient repair and scoring~\cite{OnahMichielsen2026}.

Nonuniform controls or advice mark a different resource model.  A circuit or parameter schedule without an efficient generator, or a warm-start solution supplied separately for each instance, may show attainable quality but does not establish membership in these classes.  A supplied cut can already encode a good answer, while finding effective variational parameters can itself be hard in some settings~\cite{Bittel2021}.  Classical relaxations and adaptive incumbents computed within the algorithm are allowed, with their computation and updates charged~\cite{Egger2021,LotshawIWS2026,MaciejewskiNDAWS2026}.  The family-wide fixed-angle guarantees in \Cref{sec:decoder} use a common setting rule.  The companion paper discusses the analogous advice boundary for its preparation theorem~\cite[Sec.~IV]{HadfieldApproxQOpt2026}.

Given an efficient quantum MaxCut sampler, returning the better of its cut and a deterministic \(1/2\)-approximate cut retains a \(1/2\) guarantee.  Improving this worst-case ratio requires a further bound on the quantum output.  Symmetry and locality shape fixed-depth QAOA analyses, including family-specific bounds on the \emph{decoded mean}~\cite{BravyiKlieschKoenigTang2020,ShaydulinHadfieldHoggSafro2021,MarwahaHadfield2022,FarhiGamarnikGutmann2020,BarakMarwaha2022}.  A mean upper bound limits the premise of \Cref{lem:mean-tail}, but without a tail bound it does not rule out rare high-scoring outputs that repetition could find.

Compressed encodings sharpen the distinction between quantum energy and classical output.  Pauli correlation encoding (PCE) represents binary variables by Pauli-string expectations~\cite{Sciorilli2025}, while quantum random access optimization (QRAO) uses single-qubit observables.  QRAO's optimal-energy promise problem can be quantum Merlin--Arthur (QMA)-complete on explicitly compiled MaxCut families, a classification of encoded energy rather than classical cut quality~\cite{HadfieldQRAO2026}.  The companion paper distinguishes encoded energy gain from decoded mean cut gain above the random-cut baseline~\cite[Sec.~II and Table~I]{HadfieldApproxQOpt2026}.  For standard quantum random access code (QRAC) MaxCut packings with \(d\in\{2,3\}\) variables per qubit and edge endpoints on different qubits, Fuller et al.'s local decoder transfers encoded energy gain to decoded mean cut gain by \(1/d^2\) for any input state~\cite{Fuller2024}.  The companion paper states the identity explicitly and gives a local POVM derivation~\cite[Eq.~(6) and App.~A]{HadfieldApproxQOpt2026}.  A mean-gain identity still needs a ratio bound and efficient repetition to meet our output definition.  The choice of compressed PCE encoding can leave small expectation-value margins, increasing the measurements needed to recover the encoded signs~\cite{HadfieldNoFreeCompression2026}.

\subsection{Broader problem families}
For unrestricted MaxCut and MaxSAT, the target described in \Cref{sec:separation-implications} is a uniform quantum output ratio above a randomized classical hardness threshold on the same explicit-input problem.  Known \(\NP\)-hard approximation gaps limit possible quantum ratios if \(\NP\nsubseteq\BQP\).  Restricted families can have different ratios, provided their definitions and the costs of input access are explicit.

Recent qudit QAOA work computes expected cut fractions for high-girth, regular Max-\(k\)-Cut families~\cite{ApteBoulebnaneJinOmanakuttanPerlinShaydulin2026,ApteKhanFarhiMarwahaBoulebnaneShaydulin2026}.  The later study reports that, in the degree and \(k\) regimes analyzed, QAOA at depth \(p\ge9\) and girth at least \(20\) exceeds its improved classical local-vector cut-fraction guarantee.  Because this compares QAOA with one classical algorithm on a restricted graph family, an unrestricted ratio gap would also require hardness against randomized classical algorithms at the same ratio.

Multiplicative ratios depend on objective normalization.  Adding a constant can improve a ratio without changing the optimizer, so the problems above fix their scores as part of the NPO definition.  Baseline-relative gain is useful when affine invariance is intended, but it defines a different taxonomy~\cite[Sec.~II]{HadfieldApproxQOpt2026}.

The folded-OPI result extends an oracle exact-search separation to a graded objective~\cite{YamakawaZhandry2024,KramerGilFusterJonesEisertSchreiber2026}.  The problem compendium of Hadfield et al.~\cite[App.~A]{Hadfield2019} gives QAOA mappings for MaxSAT, Set Cover, Graph Coloring, TSP, and other problems, along with then-known approximation algorithms and hardness results for many of them.  These are candidate objectives for output guarantees on every input.

\Needspace{7\baselineskip}
\subsection{Outlook}

The classical witnesses make the bounded-error quantum hierarchy strict under \(\NP\nsubseteq\BQP\), with further growing-factor distinctions in Appendix~\ref{app:growing-proofs}.  \(\CMO\) and \(\DLogFit\) instead give conditional gaps between quantum and randomized classical approximation on the same problem.  The decoded-mean result connects expectation analyses to these output classes when the decoder and its running time satisfy their definitions.

For familiar \(\APX\) problems such as MaxCut and MaxSAT, a quantum advantage would be a better guaranteed ratio, supported by classical hardness for the same instances and score normalization.  Recent OPI results prove classical-output gaps in membership queries and streaming space~\cite{KramerGilFusterJonesEisertSchreiber2026,WuYang2026,BoulandDingJain2026}.  An explicit-input time analogue, or a ratio gap for an unrestricted problem, remains open.  QAOA, DQI, and compressed encodings provide settings in which to pursue such guarantees for the final decoded solution.  Future work can identify promising problem families and regimes, then develop quantum algorithms and decoders tailored to their structure toward proving worst-case output guarantees.

\section*{Acknowledgments}

This work was supported under Air Force Research Laboratory Contract No.~\mbox{FA8750-25-C-B0040}.  Early drafts of my doctoral thesis explicitly named and defined \(\BQAPX\), \(\BQPTAS\), and \(\BQFPTAS\).  The final thesis instead reviews the classical classes and identifies quantum improvement as a goal~\cite[Sec.~5.2.1]{HadfieldThesis2018}.  The quantum definitions were omitted as we had not yet established nontrivial separations.  This paper gives precise bounded-error classical-output definitions and proves such separations.  OpenAI ChatGPT assisted with assembling subsequently developed results, consistency checks, editing, and literature synthesis, in particular with assessing the most recent DQI results discussed in Sec.~VI.A.  The author assumes full responsibility for the manuscript.

\bibliography{references}

\appendix
\onecolumngrid
\section{Glossary of complexity classes}
\label{app:class-glossary}

The approximation classes below concern classical-input NPO problems and classically checkable outputs.  Here bounded error means that, on every input, the algorithm returns a feasible solution meeting the stated ratio with probability at least \(2/3\); invalid outputs or outputs missing the ratio count as failures.  For schemes, one uniform algorithm receives the binary-encoded rational accuracy \(\epsilon\), and its running time includes the encoding length \(\ell(\epsilon)\).  The formal conventions are in \Cref{def:output} and \Cref{def:classes}.  Classical approximation terminology follows Refs.~\cite{PapadimitriouYannakakis1991,Ausiello1999,Vazirani2001}.

\begingroup
\renewcommand*\descriptionlabel[1]{\hspace\labelsep\normalfont\bfseries #1}
\begin{description}[style=multiline,leftmargin=2.8cm,labelwidth=2.6cm,itemsep=0pt,topsep=3pt]
\item[\(\mathsf{P}\)] Decision problems solvable by a deterministic polynomial-time algorithm.
\item[\(\NP\)] Decision problems whose yes instances have polynomial-length classical certificates verifiable in polynomial time.
\item[\(\mathsf{BPP}\)] Decision problems solvable by a randomized polynomial-time algorithm with error at most \(1/3\) on every input.
\item[\(\BQP\)] Decision problems solvable by a uniform polynomial-time quantum algorithm with error at most \(1/3\) on every input.
\item[\(\mathsf{QMA}\)] Decision problems whose yes instances have polynomial-size quantum witnesses accepted by a polynomial-time quantum verifier, with a constant completeness--soundness gap~\cite{AharonovAradVidick2013}.
\item[\(\FBPP\)] In search-problem statements, shorthand for a randomized polynomial-time solver returning a checkable valid answer with probability at least \(2/3\) on each input of the stated domain, which may be a promise family.
\item[\(\mathsf{NPO}\)] NP-style optimization problems with polynomial-size feasible solutions, polynomial-time instance and feasibility tests, and polynomial-time computable objective values.  We additionally require nonempty feasible sets and positive optima.
\item[\(\APX\)] NPO problems with a deterministic polynomial-time approximation at some fixed constant ratio.
\item[\(\PTAS\)] NPO problems with one deterministic scheme that, for each fixed \(\epsilon>0\), attains ratio \(1-\epsilon\) for maximization or \(1+\epsilon\) for minimization in time polynomial in the instance length.
\item[\(\FPTAS\)] \(\PTAS\) problems whose scheme runs in time polynomial jointly in the instance length, \(1/\epsilon\), and \(\ell(\epsilon)\).
\item[\(\RAPX\)] The \(\APX\) guarantee using a randomized polynomial-time classical algorithm that succeeds with probability at least \(2/3\).
\item[\(\RPTAS\)] The \(\PTAS\) guarantee using one randomized classical scheme with success probability at least \(2/3\) and polynomial instance-length time for each fixed accuracy.
\item[\(\RFPTAS\)] The randomized classical \(\PTAS\) with success probability at least \(2/3\) and time polynomial jointly in instance length, \(1/\epsilon\), and \(\ell(\epsilon)\).
\item[\(\BQAPX\)] The \(\APX\) guarantee using a uniform bounded-error quantum algorithm that returns a feasible classical solution.
\item[\(\BQPTAS\)] The \(\PTAS\) guarantee using one uniform bounded-error quantum scheme with polynomial instance-length time for each fixed accuracy and a feasible classical output.
\item[\(\BQFPTAS\)] The bounded-error quantum \(\PTAS\) with feasible classical output and time polynomial jointly in instance length, \(1/\epsilon\), and \(\ell(\epsilon)\).
\item[\(\RLOGAPX\)] Randomized classical polynomial-time approximation with factor \(O(\log(2+N))\), where \(N\) is the full input length.
\item[\(\BQLOGAPX\)] The same logarithmic-factor guarantee using a bounded-error quantum algorithm and feasible classical output.
\item[\(\RPOLYAPX\)] Randomized classical polynomial-time approximation with factor \(N^{O(1)}\).
\item[\(\BQPOLYAPX\)] The same polynomial-factor guarantee using a bounded-error quantum algorithm and feasible classical output.
\item[\(\REXPAPX\)] Randomized classical polynomial-time approximation with factor \(2^{N^{O(1)}}\). The algorithm still runs in polynomial time.
\item[\(\BQEXPAPX\)] The same exponential-factor guarantee using a bounded-error quantum algorithm and feasible classical output.
\item[\(\RAPTAS\)] A randomized classical scheme attaining a \((1+\epsilon)\) minimization ratio on inputs with optimum above a fixed \(\epsilon\)-dependent threshold, in polynomial time for every fixed accuracy.
\item[\(\mathsf{EPTAS}\)] A \(\PTAS\) with running time \(f(1/\epsilon)\operatorname{poly}(|x|)\), where the polynomial degree does not depend on \(\epsilon\).  Reading the accuracy input is also charged~\cite{CesatiTrevisan1997}.
\item[\(\mathsf{QPTAS}\)] A quasi-polynomial-time approximation scheme for each fixed accuracy.  Here the acronym is mentioned only to distinguish it from the quantum prefix ``BQ.''
\end{description}
\endgroup

The modifiers \(\NP\)-hard, \(\NP\)-complete, and QMA-hard describe reduction-based hardness, not additional approximation classes.  PCP refers to probabilistically checkable proofs.

\par\vspace*{1.2\baselineskip}
\twocolumngrid
\section{Proof of the conditional hierarchy}
\label{app:hierarchy}
\begin{proof}[Proof of \Cref{prop:hierarchy}]
For the standard Max-E3SAT variant with three distinct, noncomplementary literals per clause, the method of conditional expectation gives a deterministic \(7/8\)-approximation, so the problem lies in \(\BQAPX\).  A random assignment satisfies each clause with probability \(7/8\), and fixing its variables one at a time without decreasing the conditional expectation gives the deterministic bound.  For any fixed \(0<\eta<1/8\), it is \(\NP\)-hard to distinguish a satisfiable three-literal CNF formula with \(m\) clauses from one with optimum at most \((7/8+\eta)m\)~\cite[Theorem~6.5]{Hastad2001}.  Fix \(\epsilon<1/8-\eta\), and let \(u\mapsto F_u\) be the gap reduction from an \(\NP\)-complete language.  On input \(u\), compute \(F_u\), run the alleged \(\BQPTAS\), and accept only if its returned assignment is valid and satisfies more than \((7/8+\eta)m\) clauses.  On a YES instance, its successful assignment satisfies at least \((1-\epsilon)m>(7/8+\eta)m\) clauses.  On a NO instance, no assignment can cross the threshold.  Classical verification therefore gives a bounded-error quantum decider with one-sided soundness for the source language, implying \(\NP\subseteq\BQP\).  Thus Max-E3SAT is not in \(\BQPTAS\) under the stated assumption.

Restrict unweighted Planar Vertex Cover to graphs with at least one edge.  This does not change its exact hardness.  It has a classical \(\PTAS\) and hence lies in \(\BQPTAS\)~\cite{Baker1994}.  On an \(n\)-vertex instance, \(1\le\OPT\le n\).  A \(\BQFPTAS\) called with \(\epsilon=1/(2n)\) returns, on its successful event, a cover \(C\) with
\[
 |C|\le\left(1+\frac{1}{2n}\right)\OPT
 \le \OPT+\frac12<\OPT+1.
\]
Integrality makes \(C\) optimal.  Since exact Planar Vertex Cover is \(\NP\)-hard~\cite{GareyJohnsonStockmeyer1976}, this optimizer gives a bounded-error quantum decision algorithm for an \(\NP\)-complete language, again implying \(\NP\subseteq\BQP\).  Hence Planar Vertex Cover is not in \(\BQFPTAS\) under the assumption.
\end{proof}

\Needspace{12\baselineskip}
\section{Further approximation classes and classical witnesses}
\label{app:growing-proofs}

The definitions and proofs for \Cref{tab:extended-witnesses} follow.

\subsection{Beyond fixed ratios}
\label{sec:beyond-ratios}

\begin{definition}[Growing-factor classes]
\label{def:growing-factors}
Write \(N=|x|\) for the full bit length of the explicit input, and take the factor \(\rho(N)\ge1\) to be nondecreasing.  A \(\rho\)-approximation succeeds with probability at least \(2/3\) and returns a feasible solution of value at least \(\OPT(x)/\rho(N)\) for maximization, or cost at most \(\rho(N)\OPT(x)\) for minimization.  With uniform polynomial-time randomized classical algorithms, the factors \(O(\log(2+N))\), \(N^{O(1)}\), and \(2^{N^{O(1)}}\) define \(\RLOGAPX\), \(\RPOLYAPX\), and \(\REXPAPX\), respectively.  With uniform polynomial-time quantum algorithms producing feasible classical outputs, they define \(\BQLOGAPX\), \(\BQPOLYAPX\), and \(\BQEXPAPX\).  Constants and exponents may depend on the problem but not on its instance.  The conventions for invalid outputs and total running time are those of \Cref{def:output}.
\end{definition}

Here \(\mathrm{EXP}\) describes an approximation \emph{factor}, while all algorithms still run in polynomial time.  The inclusions from \(\APX\) through these classes hold in each model, and each randomized class lies in its quantum counterpart.

\begin{proposition}[Conditional growing-factor hierarchy]
\label{prop:growing-hierarchy}
If \(\NP\nsubseteq\BQP\), then
\[
\begin{aligned}
  \BQAPX&\subsetneq\BQLOGAPX\subsetneq\BQPOLYAPX\\
  &\subsetneq\BQEXPAPX.
\end{aligned}
\]
The corresponding chain with \(\mathsf{R}\) in place of \(\mathsf{BQ}\) is also strict.
\end{proposition}

Greedy covering and an \(\NP\)-hard logarithmic gap place Set Cover between the logarithmic and constant-factor classes~\cite{Chvatal1979,DinurSteurer2014}.  A single vertex and an \(n^{1-\delta}\) hardness gap place Max Clique between polynomial and logarithmic factors~\cite{Zuckerman2007}.  Binary positive-weight TSP has an exponential upper factor but no polynomial factor under the same assumption.  The proofs account for full input bit length.

\begin{proposition}[An asymptotic guarantee need not be an absolute \(\PTAS\)]
\label{prop:bin-packing}
Define \(\RAPTAS\) for minimization as in \(\RPTAS\), except that, for each fixed \(\epsilon\), its \((1+\epsilon)\) bound need only hold when \(\OPT(x)\ge T_\epsilon\) for a constant \(T_\epsilon\).  Every run must still produce a feasible packing or declare failure.  Bin Packing on nonempty lists of positive rational item sizes at most one lies in \(\RAPTAS\cap\RAPX\).  If \(\NP\nsubseteq\BQP\), it does not lie in \(\BQPTAS\).
\end{proposition}

The scheme of Fernandez de la Vega and Lueker uses at most \((1+\epsilon)\OPT+c_\epsilon\) bins, with \(c_\epsilon\) depending only on accuracy~\cite{FernandezDeLaVegaLueker1981}.  Running it at \(\epsilon/2\) gives a \((1+\epsilon)\) ratio for \(\OPT\ge T_\epsilon=\lceil2c_{\epsilon/2}/\epsilon\rceil\), because then \(c_{\epsilon/2}\le(\epsilon/2)\OPT\).  The proof gives a constant-factor algorithm and the \(\NP\)-hard gap between two bins and at least three bins.  Because bin counts are integral, the additive term does not imply a \(\PTAS\).

Rectilinear TSP on distinct integer-grid points is a geometric alternative to Planar Vertex Cover.  Its Manhattan (\(L_1\)) tour cost is integral and exactly computable, and it has a classical \(\PTAS\) while remaining \(\NP\)-hard on polynomial-grid constructions~\cite{Arora1998,GareyGrahamJohnson1976}.  Thus it lies in \(\RPTAS\setminus\BQFPTAS\) under the same assumption.  Unlike unrounded Euclidean distance, its cost directly meets our exact rational-score NPO convention.

\subsection{Proofs}

\begin{proof}[Proof of \Cref{prop:growing-hierarchy}]
For Set Cover restrict to explicitly listed, nonempty universes covered by the input sets.  Greedy covering uses at most \(H_m\OPT\) sets when the universe has \(m\) elements, where \(H_m=\sum_{j=1}^m1/j\) is the \(m\)th harmonic number.  It satisfies \(H_m\le1+\ln m=O(\log(2+N))\), where \(N\) is the full input length~\cite{Chvatal1979}.  Dinur and Steurer give a deterministic \(\NP\)-hard gap excluding a \((1-\delta)\ln m\) approximation for every fixed \(\delta>0\)~\cite{DinurSteurer2014}.  Their hardness factor concerns universe size \(m\), while our class factor concerns full input length \(N\).  The growing gap in \(m\) already excludes every fixed factor.  The deterministic reduction preserves the stated assumption \(\NP\nsubseteq\BQP\).  If a bounded-error quantum algorithm gave a constant factor, run it on a gap-reduction output and compare the size of its verified cover to the gap threshold.  In the YES case it crosses that threshold with probability at least \(2/3\).  In the NO case no cover can cross.  This would put an \(\NP\)-complete language in \(\BQP\).  Thus Set Cover is in \(\RLOGAPX\setminus\BQAPX\).

For Max Clique on graphs with \(n\ge1\) vertices, any vertex gives a clique of size one, hence factor at most \(n\le N\).  Zuckerman's deterministic \(\NP\)-hard gap excludes ratios \(n^{1-\delta}\) for each fixed \(\delta>0\)~\cite{Zuckerman2007}.  Taking \(\delta=1/2\) and using the explicit graph encoding \(N=n^{O(1)}\), we have \(O(\log(2+N))=o(n^{1/2})\).  A hypothetical \(\BQLOGAPX\) algorithm, followed by classical clique checking and a size comparison, would solve that gap with bounded quantum error.  Consequently Max Clique lies in \(\RPOLYAPX\setminus\BQLOGAPX\).

For Positive-Weight TSP, take complete undirected graphs on at least three vertices with positive integer edge weights in binary.  The large-nonedge-weight Hamiltonian-cycle gap below is the classical Sahni--Gonzalez construction~\cite{SahniGonzalez1976}.  Every tour has cost at most \(nW\) when \(W\) is the maximum weight, and \(\OPT\ge n\).  Thus an arbitrary tour has factor at most \(W\le2^N\).  Suppose a \(\BQPOLYAPX\) algorithm has factor at most \(CN^c\) for fixed constants \(C,c\).  Given a graph \(G\) with \(n\) vertices, construct the TSP instance with weight one for edges of \(G\) and weight \(M=n^K\) for nonedges, with the fixed integer \(K>2c+3\).  Its encoding length is \(N=O(n^2\log n)\), with a constant depending on \(K\).  If \(G\) has a Hamiltonian cycle, \(\OPT=n\).  Otherwise every tour has cost at least \(M+n-1\).  Since \(CN^c n=O(n^{2c+1}(\log n)^c)\), our choice of \(K\) ensures \(CN^c n<M\) for all sufficiently large \(n\).  Run the proposed algorithm, verify the tour, and accept precisely if its cost is below \(M\).  This decides Hamiltonian Cycle with bounded quantum error, contradicting \(\NP\nsubseteq\BQP\).  Finitely many smaller inputs can be decided by a fixed lookup.  Hence Positive-Weight TSP is in \(\REXPAPX\setminus\BQPOLYAPX\).

These deterministic upper bounds give the quantum containments and also the randomized classical containments.  Since each excluded quantum class contains its randomized classical counterpart, the same three witnesses prove strictness of the randomized classical chain.
\end{proof}

\begin{proof}[Proof of \Cref{prop:bin-packing}]
The asymptotic scheme gives the \(\RAPTAS\) membership~\cite{FernandezDeLaVegaLueker1981}.  Next Fit places an item in the current bin if it fits and otherwise opens a new bin.  The first item of each new bin did not fit in its predecessor, so the loads of these two bins sum to more than one.  If Next Fit uses \(b\) bins, pairing consecutive bins shows that the total item size is greater than \(\lfloor b/2\rfloor\), while it is at most \(\OPT\).  Integrality gives \(b\le2\OPT\), so Bin Packing belongs to \(\RAPX\).

For the exclusion, take a Partition instance of positive integers summing to \(2B\), with each integer at most \(B\), and make item sizes \(a_i/B\).  They fit in two unit bins exactly when the integers split into two groups of sum \(B\).  Otherwise every packing uses at least three bins.  Suppose a \(\BQPTAS\) existed and set \(\epsilon=1/4\).  On a YES instance its successful output uses at most \(2(1+\epsilon)=5/2\) bins, and therefore at most two.  On a NO instance no two-bin output exists.  Verifying the packing gives a bounded-error quantum decision procedure for Partition, implying \(\NP\subseteq\BQP\).
\end{proof}

\begin{proof}[Rectilinear TSP as an alternative hierarchy witness]
Use integer coordinates, at least three distinct points, and the exact \(L_1\) length of a Hamiltonian tour.  Arora's \(\PTAS\) applies to this geometric norm, so the problem lies in \(\RPTAS\)~\cite{Arora1998}.  Rectilinear TSP is \(\NP\)-hard on polynomially bounded integer-grid constructions~\cite{GareyGrahamJohnson1976}.  On an instance of such a construction, compute the integral cost \(U\) of any fixed-order tour.  The construction has polynomially many points with polynomially bounded integer coordinates.  Thus \(U\) and \(1/\epsilon=2U\) are polynomial in the source length.  If a \(\BQFPTAS\) existed, calling it with the rational accuracy \(\epsilon=1/(2U)\) would take polynomial time and, on its successful event, return a tour of cost at most
\[
 (1+\epsilon)\OPT\le\OPT+\tfrac12.
\]
Integral tour costs force this tour to be optimal.  Exact comparison with the reduction's threshold would then decide an \(\NP\)-complete language in \(\BQP\).
\end{proof}

\section{Maximum-order proof details}
\label{app:cmo-proofs}
\subsection*{Factoring from an exponent}
\begin{proof}[Proof of \Cref{lem:exponent-factor}]
Write \(N=\prod_{i=1}^d p_i^{e_i}\), where \(d\ge2\), and write \(M=2^tu\) with \(u\) odd.  Sample \(b\) uniformly from \(\{1,\ldots,N-1\}\).  If \(1<\gcd(b,N)<N\), this gcd is already a factor.  Otherwise condition on \(b\) being a unit.  The same routine accepts an arbitrary positive trial exponent: for a unit, it reports failure if \(b^M\not\equiv1\pmod N\), and it returns a gcd only when that gcd is a verified nontrivial factor.  For the multiples of \(\lambda(N)\) covered by the lemma, the initial exponent check always passes.  By the Chinese remainder theorem its component projections are independent and uniform.

Let the Sylow-2 subgroup, the subgroup of 2-power order, of \(U_{p_i^{e_i}}\) have order \(2^{s_i}\).  Here \(\varphi(m)\) is Euler's totient, the number of residues coprime to \(m\).  Each \(p_i\) is odd, so \(\varphi(p_i^{e_i})=p_i^{e_i-1}(p_i-1)\) is even and \(s_i\ge1\).  Because \(M\) is a multiple of \(\lambda(N)\), \(t\ge\max_i s_i\) and \(u\) is divisible by every component exponent's odd part.  Raising to \(u\) therefore kills the odd-order factor.  Because \(u\) is odd, it also permutes the cyclic Sylow-2 subgroup, making \(b^u\) uniform on it.  Put
\(
 L_i=\log_2\ord_{p_i^{e_i}}(b^u).
\)
For \(0\le \ell\le s_i\),
\[
 \Pr[L_i=0]=2^{-s_i},
 \qquad
 \Pr[L_i=\ell]=2^{\ell-1-s_i}\quad(\ell\ge1).
\]
Indeed, a cyclic group of order \(2^{s_i}\) has one identity and \(2^{\ell-1}\) elements of order \(2^\ell\).
For example, if \(s_1\le s_2\), independence gives
\[
 \begin{aligned}
 \Pr[L_1=L_2]
 &=2^{-s_1-s_2}\left(1+\sum_{\ell=1}^{s_1}4^{\ell-1}\right)\\
 &=\frac{4^{s_1}+2}{3\cdot2^{s_1+s_2}}\le\frac12.
 \end{aligned}
\]
Relabel so that these are the two coordinates with smallest \(s_i\), in nondecreasing order.  If all \(L_i\) coincide, then \(L_1=L_2\).  Their independence bounds the probability of this event by \(1/2\), irrespective of the remaining components.
If the \(L_i\) are not all equal, then the squaring sequence \(b^u,b^{2u},\ldots,b^{2^tu}\), whose last term is \(b^M\equiv1\pmod N\),
contains a square root \(z\) of one whose component signs are not all equal.  Just before the last nontrivial term is squared, components with the largest \(L_i\) equal \(-1\), while those with smaller \(L_i\) equal \(+1\).  Hence \(z\not\equiv\pm1\pmod N\), and \(\gcd(z-1,N)\) is a nontrivial factor.  The gcd branch succeeds immediately on nonunits, and the conditional probability of unequal component orders on units is at least \(1/2\).  This is the standard Miller--Shor exponent-to-factoring procedure~\cite{Miller1976,Shor1997,Ekera2021}.
\end{proof}

\subsection*{Certified Maximum Order}
\begin{proof}[Proof of \Cref{thm:cmo}]
For the quantum upper bound, use amplified Shor calls to factor \(N\) and every relevant \(p-1\), classically verifying every returned prime factorization.  Write \(N=\prod_{i=1}^d p_i^{e_i}\) with distinct primes \(p_i\).  For an odd prime-power component \(p^e\), the unit group is cyclic of order \(\varphi(p^e)=p^{e-1}(p-1)\).  Its known factorization lets us sample and certify a generator by testing the prime divisors of \(\varphi(p^e)\).  Specifically, a unit \(a\) is a generator when \(a^{\varphi(p^e)/\ell}\not\equiv1\pmod{p^e}\) for each distinct prime \(\ell\mid\varphi(p^e)\).  Among units, the generator probability is
\(\varphi(\varphi(p^e))/\varphi(p^e)=\Omega(1/\max\{1,\log\log N\})\)~\cite{RosserSchoenfeld1962}.  Indeed, apply the standard lower bound for \(\varphi(m)/m\) to \(m=\varphi(p^e)\le p^e\le N\), absorbing the finitely many small values into the constant.  Use enough certified trials on each of the at most \(\log_2N\) components to make its failure probability at most \(1/[20(1+\log_2N)]\), and amplify the polynomially many verified factorization calls so their total failure probability is at most \(1/20\).  The union bound then gives success probability at least \(9/10\) in worst-case polynomial time.  For \(2^e\), use the trivial component for \(e=1\), the generator \(3\bmod4\) for \(e=2\), and \(5\), of order \(2^{e-2}\), for \(e\ge3\).  By the Chinese remainder theorem, combining maximum-order elements for the prime-power components produces an element whose order is their least common multiple,
\[
  \operatorname{lcm}_{1\le i\le d}\lambda(p_i^{e_i})=\lambda(N).
\]
This includes prime and prime-power inputs, for which the Chinese remainder step has only one component.
Merging the known component factorizations supplies the prime factorization of \(\lambda(N)\) required by the certificate.  If a truncated subroutine fails, output \((1,1,\varnothing)\).  Otherwise the output is optimal with probability at least \(2/3\).  The algorithm ignores \(\epsilon\), so it is a \(\BQFPTAS\).

For the converse, take a composite input to factoring.  Return \(2\) if it is even, and test integer roots of degrees \(2\) through \(\lfloor\log_2N\rfloor\), returning a nontrivial base if \(N\) is a perfect power.  The remaining input is odd and has at least two distinct prime factors.  Write \(n\) for the binary length of \(N\).  Increase the fixed constants in any polynomial-factor bound to integers \(C\ge1\) and \(c\ge0\).  Validate the output.  On the probability-at-least-\(2/3\) successful event it supplies a certified order \(r\ge\lambda(N)/(Cn^c)\).  Every element order divides the group exponent, hence
\[
  k_0=\frac{\lambda(N)}{r}
\]
is an integer with \(1\le k_0\le Cn^c\).  The constants are hard-coded into the reduction.  For each \(k\in\{1,\ldots,Cn^c\}\), run the explicit gcd-and-squaring routine from \Cref{lem:exponent-factor} on \(M=kr\).  The routine is defined for every \(M\), may report failure, and returns only verified nontrivial factors.  All trial exponents have polynomial bit length, and there are polynomially many trials.  At \(k=k_0\), the lemma gives conditional success probability at least \(1/2\).  Thus one outer run succeeds with probability at least \(1/3\), and constant independent repetition raises the probability above \(2/3\).  If factoring itself is in \(\FBPP\), the exact classical construction above shows that \(\CMO\in\RPOLYAPX\) as well.
\end{proof}

\section{Repetition counts for fixed-depth MaxCut QAOA}
\label{app:qaoa-repetitions}

For the QAOA guarantees discussed in \Cref{sec:decoder}, let \(p\) be the circuit depth, \(Y\) a measured cut, \(C(Y)\) its score, and \(\mu_p\) a family-wide lower bound on the mean approximation ratio.  A uniformly implementable circuit with \(\mathbb E[C(Y)]/\OPT\ge\mu_p\) on every graph of its family has, by \Cref{lem:mean-tail}, probability at least \((\mu_p-\tau)/(1-\tau)\) of exceeding \(\tau\OPT\) when \(0\le\tau<\mu_p<1\).  The chance that all \(k\) independent cuts miss this threshold is at most \(((1-\mu_p)/(1-\tau))^k\).  Taking the best of \(k_p(\tau)=\bigl\lceil\log(1/3)/\log((1-\mu_p)/(1-\tau))\bigr\rceil\) independent outputs succeeds with probability at least \(2/3\).  One run suffices for \(\mu_p=1\).  The full running-time convention of \Cref{sec:classes} applies.

\section{Discrete-Logarithm Fitting as a search-padding example}
\label{app:keyfit}

This example has the one-half threshold of \Cref{prop:search-padding} with an empirical-accuracy objective.  Its score is defined on an explicit sample rather than by the two-level score of the generic construction.

Let \(p=2q+1\) be prime with \(q\) prime (so \(p\) is a safe prime), let \(G\) be the order-\(q\) subgroup of \(\mathbb Z_p^*\), and let \(g\) generate \(G\).  We write \(\mathbb Z_q\) for exponents modulo \(q\).  The input group element is \(h=g^s\) for an unknown \(s\in\mathbb Z_q\).  On this promise family, \(\DLog\) is the search problem of recovering \(s\) from \((p,q,g,h)\).  Validity is recognized using primality tests and the conditions \(g^q\equiv h^q\equiv1\pmod p\), \(g\ne1\), and \(1\le g,h<p\).  Because \(q\) is prime, these conditions ensure that \(g\) has order \(q\) and \(h\in G\).  A labeled feature has the form \((r,z)\), where
\[
  z\equiv g^b h^r\pmod p,
  \qquad b\in\{0,1\},\quad r\in\mathbb Z_q.
\]
This is the nonredundant part of an ElGamal-style ciphertext~\cite{ElGamal1985}.  Because \(r\) is explicit, the usual first component \(g^r\) carries no additional information.  The label \(b\) itself is easy to read by comparing \(z\) with \(h^r\) and \(gh^r\).  The search task is to find an exponent hypothesis that fits the sample, not to decode an individual label.

\begin{definition}[Discrete-Logarithm Fitting problem]
An instance of \(\DLogFit\) consists of primes \(q\) and \(p=2q+1\), canonical residues \(1\le g,h<p\) satisfying \(g\ne1\) and \(g^q\equiv h^q\equiv1\pmod p\), and a nonempty list
\[
  S=((r_i,z_i,b_i))_{i=1}^m
\]
with \(m\ge1\), \(0\le r_i<q\), \(1\le z_i<p\), \(b_i\in\{0,1\}\), and
\[
 z_i\equiv g^{b_i}h^{r_i}\pmod p
\]
for every \(i\).  A feasible hypothesis is a tagged exponent \((\mathsf{exp},a)\), where \(a\in\mathbb Z_q\), or one of the tagged constants \(\mathsf{zero}\) and \(\mathsf{one}\).  The exponent hypothesis predicts
\[
 P_a(r,z)=
 \begin{cases}
 0,&z\equiv g^{ar}\pmod p,\\
 1,&z\not\equiv g^{ar}\pmod p,
 \end{cases}
\]
and its value is the fraction of the \(m\) labels \(b_i\) predicted correctly.  The constants always predict zero or one.
\end{definition}

Hypotheses have \(O(\log q)\) bits, and instance recognition, feasibility, and score evaluation are polynomial in the explicit list length using primality tests and modular exponentiation.  The secret exponent \(s\) scores one on every valid instance, so \(\OPT=1\), while the better constant predictor scores at least \(1/2\).  Thus the problem is in NPO and in \(\RAPX\).

\begin{theorem}[DLog-Fit threshold]
\label{thm:keyfit}
\(\DLogFit\) has a deterministic \(1/2\)-approximation and a bounded-error quantum algorithm that returns an optimum.  For every fixed rational \(0<\delta<1/2\), a randomized polynomial-time \((1/2+\delta)\)-approximation yields a randomized polynomial-time discrete-logarithm algorithm on valid safe-prime subgroup instances \((p,q,g,h)\).  Hence, if \(\DLog\notin\FBPP\) on this promise family, then
\[
\begin{gathered}
  \DLogFit\in\RAPX\cap\BQFPTAS,\\
  \DLogFit\notin\RPTAS.
\end{gathered}
\]
On the same promise family, \(\DLogFit\in\RPTAS\) if and only if
\(\DLogFit\in\RFPTAS\), and both conditions hold if and only if
\(\DLog\in\FBPP\).
\end{theorem}

\begin{proof}
The deterministic baseline chooses the empirical majority label.  Shor's discrete-logarithm algorithm recovers \(s\), whose hypothesis classifies every example correctly~\cite{Shor1997}.  Repeat the quantum subroutine a fixed number of times and check each candidate \(\hat s\in\mathbb Z_q\) by computing whether \(g^{\hat s}\equiv h\pmod p\).  Since \(g\) has order \(q\), a candidate passes exactly when \(\hat s=s\).  If all truncated trials fail, return the majority constant.  This is always feasible and is optimal with bounded error.

For the converse, generate a training list with independent uniform \(r_i\in\mathbb Z_q\) and \(b_i\in\{0,1\}\), setting \(z_i\equiv g^{b_i}h^{r_i}\pmod p\).  Fix \(a\ne s\).  Conditional on \(b=0\), the congruence \(z\equiv g^{ar}\pmod p\) holds only for \(r=0\), so \(P_a\) is correct with probability \(1/q\).  Conditional on \(b=1\), it holds for the unique solution of
\[
  1+(s-a)r=0\pmod q,
\]
which is precisely the value on which \(P_a\) is wrong.  Its conditional accuracy is \(1-1/q\).  Averaging over the label gives population accuracy exactly \(1/2\).  The two constant hypotheses also have population accuracy \(1/2\).

There are \(q-1\) incorrect exponent hypotheses and two constant hypotheses, hence \(q+1\) nonsecret hypotheses.  For each fixed hypothesis, the sample correctness indicators are independent and lie in \([0,1]\).  Hoeffding's inequality bounds the chance that its empirical accuracy reaches \(1/2+\delta\) by \(e^{-2m\delta^2}\).  Choose
\[
  m\ge \frac{\ln(12(q+1))}{2\delta^2},
\]
so a union bound limits the probability of any nonsecret hypothesis reaching the threshold to \((q+1)e^{-2m\delta^2}\le1/12\).  With probability at least \(11/12\), all such hypotheses fall below it.  Since \(\delta\) is fixed, this list has polynomial length.  One reduction trial generates a fresh list, runs the alleged approximation algorithm once, rejects either constant tag, and accepts an exponent \(a\) only if \(g^a\equiv h\pmod p\).  On the concentration event every successful approximation output is \(s\).  The approximation guarantee applies to each generated list separately, so averaging over lists preserves its conditional success probability of at least \(2/3\).  One trial therefore recovers \(s\) with probability at least \((11/12)(2/3)=11/18\), with no false positive.  Two independent trials succeed with probability at least \(1-(7/18)^2=275/324>2/3\), giving an \(\FBPP\) discrete-logarithm algorithm.

If safe-prime discrete logarithm is in \(\FBPP\), use that algorithm to recover and verify \(s\), then output its exponent hypothesis.  On failure output the majority constant.  This is a randomized polynomial-time exact optimizer on its successful event and belongs to \(\RFPTAS\), proving the reverse implication.
\end{proof}

A \(\PTAS\) called at \(\epsilon=1/4\) contradicts the theorem.  Every labeled feature can be generated from \((g,h)\), and \(\OPT=1\) on valid instances.  This is a thin optimization encoding of discrete-logarithm search, related to cryptographic learning arguments but with a different task and access model~\cite{KearnsValiant1994,ServedioGortler2004,LiuArunachalamTemme2021,SwekeSeifertHangleiterEisert2021}.  The theorem makes no claim about improvements of size \(o(1)\).

\end{document}